\documentclass[11pt,letterpaper]{article}
\usepackage[utf8]{inputenc}
\usepackage{times}
\usepackage[ruled,vlined,commentsnumbered,titlenotnumbered]{algorithm2e}
\usepackage{url}
\usepackage{fullpage}
\usepackage{pslatex} 
\usepackage{mathdots} 
\usepackage{amsmath}
\usepackage{algorithmic}
\usepackage{amssymb}
\usepackage{amsthm}
\usepackage{color}
\usepackage{array}
\usepackage{xy}
\usepackage{setspace}
\usepackage{varwidth}
\usepackage{multicol}
\usepackage{hyperref}
\usepackage{cite}
\usepackage[margin=1in,letterpaper]{geometry}
\usepackage{accents}

\newcommand{\defn}[1]{\emph{\textbf{{#1}}}}

\usepackage{todonotes}

\newcommand{\R}{\mathbb{R}}

\renewcommand{\paragraph}[1]{\vspace{.5 cm} \noindent \textbf{#1} }
\renewcommand{\vec}[1]{#1}

\usepackage{thmtools}
\usepackage{thm-restate}

\newtheoremstyle{slanted}
{3pt}
{3pt}
{\slshape}
{}
{\bfseries}
{.}
{.5em}
{}
\theoremstyle{slanted}
\newtheorem{theorem}{Theorem}
\newtheorem{lemma}[theorem]{Lemma}

\newtheorem{corollary}[theorem]{Corollary}

\begin{document}

\date{}

\title{A Nearly Tight Lower Bound for the $d$-Dimensional Cow-Path Problem}
\author{Nikhil Bansal\footnote{University of Michigan,  \url {bansal@gmail.com}. Supported in part by the NWO VICI grant 639.023.812} \phantom{aa} John Kuszmaul\footnote{Yale University, \url{john.kuszmaul@yale.edu}. Supported by NSF CCF-2106827.} \phantom{aa} William Kuszmaul\footnote{MIT CSAIL, \url{kuszmaul@mit.edu}.  Funded by a Fannie and John Hertz Fellowship and an NSF GRFP Fellowship. This research was also partially sponsored by the United States Air Force Research Laboratory and the United States Air Force Artificial Intelligence Accelerator and was accomplished under Cooperative Agreement Number FA8750-19-2-1000. The views and conclusions contained in this document are those of the authors and should not be interpreted as representing the official policies, either expressed or implied, of the United States Air Force or the U.S. Government. The U.S. Government is authorized to reproduce and distribute reprints for Government purposes notwithstanding any copyright notation herein.}}
\maketitle
\thispagestyle{empty}

\begin{abstract}
In the $d$-dimensional cow-path problem, a cow living in $\mathbb{R}^d$ must locate a $(d - 1)$-dimensional hyperplane $H$ whose location is unknown. The only way that the cow can find $H$ is to roam $\mathbb{R}^d$ until it intersects $\mathcal{H}$. If the cow travels a total distance $s$ to locate a hyperplane $H$ whose distance from the origin was $r \ge 1$, then the cow is said to achieve competitive ratio $s / r$.

It is a classic result that, in $\mathbb{R}^2$, the optimal (deterministic) competitive ratio is $9$. In $\mathbb{R}^3$, the optimal competitive ratio is known to be at most $\approx 13.811$. But in higher dimensions, the asymptotic relationship between $d$ and the optimal competitive ratio remains an open question. The best upper and lower bounds, due to Antoniadis et al., are $O(d^{3/2})$ and $\Omega(d)$, leaving a gap of roughly $\sqrt{d}$. In this note, we achieve a stronger lower bound of $\tilde{\Omega}(d^{3/2})$.
\end{abstract}
\vfill
\pagebreak

\newpage 
\pagenumbering{arabic}

\section{Introduction}

The cow-path problem is one of the simplest algorithmic problems taught to undergraduates: A cow begins at the origin on the number line, and must find a hay-stack located at some unknown point $p \in \mathbb{R}$ satisfying $|p| \ge 1$. How should the cow go about locating point $p$, if the cow wishes to optimize its worst-case competitive ratio, which is given by $s / |p|$ where $s$ is the total distance traveled by the cow?

The optimal solution is to perform a repeated doubling argument: the cow travels between the points $(-2)^0, (-2)^1, (-2)^2, \ldots$. This path can be shown to have a competitive ratio of 9, which is optimal for any deterministic solution \cite{beck1970yet} -- more generally, the class of paths that achieve competitive ratio 9 has been studied in great detail \cite{angelopoulos2018best}. Randomized solutions can do even better, achieving a competitive ratio of $\approx 4.591$. \cite{beck1970yet}.

In the decades since it was first introduced, the cow path problem has been generalized in many natural ways: to consider paths in a multi-lane highway \cite{kao1996searching, kao1998optimal}, paths that search for various type of objects \cite{baeza1995parallel},  etc. For a (slightly out of date) survey on these types of problems, see \cite{gal}.

Perhaps surprisingly, however, one of the most natural generalizations still remains quite enigmatic. In the \defn{$d$-dimensional cow-path problem}, the cow begins at the origin in $\mathbb{R}^d$, and must travel in search of a $(d - 1)$-dimensional hyperplane. Once the cow has intersected the hyperplane, their path is complete. Even in $d = 2$ dimensions, the optimal competitive ratio remains an open question -- it is conjectured to $\approx 13.811$, and to be achieved by a logarithmic spiral \cite{baeza1995parallel, finch2005searching}. In higher dimensions, even the \emph{asymptotic behavior} of the optimal competitive ratio remains open. The best upper and lower bounds, due to Antoniadis et al. \cite{antoniadis2022online}, are $O(d^{3/2})$ and $\Omega(d)$, leaving a gap of roughly $\sqrt{d}$. 

In this note, we settle the optimal $d$-dimensional competitive ratio up to low-order terms, presenting a simple argument for a lower bound of $\tilde{\Omega}(d^{3/2})$, or, more precisely, $\Omega(d^{3/2} / \sqrt{\log d})$. 

Concurrent work by Ghomi and Wenk \cite{GhomiWe22}, posted on arXiv a few days before this note, establishes a $d$-dimensional competitive ratio of $\Omega(d^{3/2})$ (communicated by Nazarov). The earlier work of Ghomi and Wenk \cite{ghomi2021shortest} considered the special case of $d=3$.

\section{Preliminaries}

In the \defn{$d$-dimensional cow-path problem}, a cow starts at $\vec{0}$ in $\mathbb{R}^d$, and wishes to find a $(d-1)$-dimensional hyperplane $H$ whose distance $r \ge 1$ from $\vec{0}$ is unknown.
The cow travels along a path until the cow intersects $H$. At this point, if the cow has traveled a total distance of $s$, then the cow is said to have achieved a competitive ratio of $\frac{s}{r}$. In general, a cow path is said to achieve competitive ratio $\alpha$ if its competitive ratio is bounded above by $\alpha$ for all $(d - 1)$-dimensional hyperplanes $H$ that are of distance at least $1$ from the origin.

Up to a constant factor in the optimal competitive ratio, using the standard doubling argument, one can assume without loss of generality that $H$ has distance $1$ from $\vec{0}$ \cite{antoniadis2022online}. Defining $\mathcal{H}$ to be the set of such hyperplanes, the competitive ratio of the cow path is simply the total distance traveled until the cow path has hit every $H \in \mathcal{H}$. 

As shown in \cite{antoniadis2022online}, using the polar duality between points and hyperplanes, this latter problem becomes equivalent to the following \defn{sphere-inspection problem}. Let $U$ be the unit sphere in $\R^d$ centered at the origin. We say that a point $p \in \mathbb{R}^d$, satisfying $|p| \ge 1$ where $|\cdot|$ denotes the standard Euclidean norm, \defn{views} a point $q \in U$ if the segment $\overline{pq}$ intersects $U$ only at $q$. (If $|p| < 1$, it does not view any points in $U$.) 
One can show that a path $P$ intersects every hyperplane $H \in \mathcal{H}$ if and only if each point $q \in U$ is visible from some point $p \in P$ \cite{antoniadis2022online}. This is because in order for the path $P$ (starting from the origin) to view a point $q \in U$, the path must intersect the hyperplane $H \in \mathcal{H}$ that is orthogonal to $\vec{q}$.   

Thus, the main result of this note can be reformulated as follows: any path $P$ that starts at the origin and views every point in $U$ must have length at least $\Omega(d^{3/2} / \sqrt{\log d})$. This then implies a $\Omega(d^{3/2} / \sqrt{\log d})$ lower bound for the optimal competitive ratio of the $d$-dimensional cow-path problem.

It is worth commenting on several other ways to think about the same problem \cite{antoniadis2022online}: A path $P$ views all of $U$ if and only if $U$ is contained in the convex hull of $P$. Also, a point $p \in P$ views a point $q \in U$ if and only if the angle $\angle(0, q, p)$ is greater than or equal to $90^\circ$---this, in turn, is equivalent to $\langle p-q,0-q\rangle \le 0$, which can be rewritten as $\langle p,q\rangle \geq |q|^2 =1$. Both the interpretations of sphere visibility (the convex-hull interpretation and the $\langle p, q \rangle \ge 1$ interpretation) will useful throughout our proofs.

We remark that, although the sphere-inspection problem and the $d$-dimensional cow path problem are asymptotically equivalent, their optimal competitive ratios differ within constant factors. Indeed, the exact optimal competitive ratio for the sphere-inspection problem is known for $d \le 3$ \cite{ghomi2021shortest}, while the optimal bound for the cow path problem remains open even for $d = 2$ \cite{baeza1995parallel, finch2005searching, antoniadis2022online}.

\section{The Lower Bound}

We will use $d$ as a global variable, indicating that we are in $\mathbb{R}^d$. And we will use $U= \{x \in \mathbb{R}^d \mid |x| = 1\}$ to denote the unit sphere.  We begin by recalling a standard result on the distribution of $U$'s mass.

\begin{lemma}[Lemma 2.2 of \cite{ball1997elementary}]
Let $\vec{v}$ be any point satisfying $|v| = 1$. The fraction of points $u \in U$ satisfying $\langle u, v \rangle \ge \epsilon$ is at most $e^{-d \epsilon^2 / 2}$. 
\label{lem:chernoff}
\end{lemma}

Lemma \ref{lem:chernoff} directly gives the following bound on the fraction of $U$ that is visible from any given point $p$.
\begin{lemma}
Let $\vec{p}$ be any point of size $|\vec{p}| \ge 1$. The fraction of points on the surface of $U$ that are visible from point $p$ is at most
$e^{-d / (2 |\vec{p}|^2)}$.
\label{lem:pointvis}
\end{lemma}
\begin{proof}
Recall that a point $q$ is visible from $p$ iff $\langle p,q\rangle \geq 1$. Thus $\langle q,\hat{p}\rangle \geq 1/|p|  $ where $\hat{p}$ is the unit vector along $p$. Applying Lemma \ref{lem:chernoff} with $v=\hat{p}$ and $\epsilon=1/|p|$ gives the result.
%
%
\end{proof}

\begin{figure}
    \centering
    \includegraphics{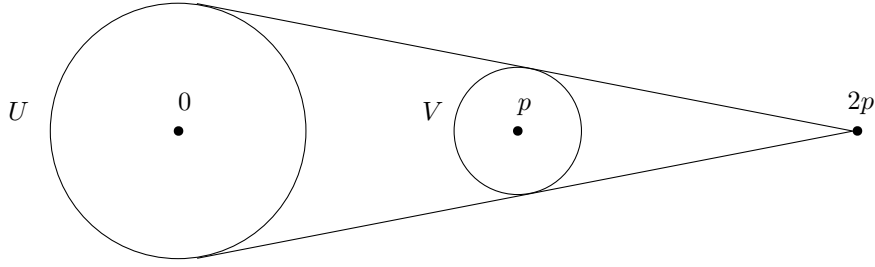}
    \caption{A visual illustration of how the points $q \in U$ visible from $V$ are also visible from point $2p$.}
    \label{fig:cones}
\end{figure}

Next we consider the fraction of $U$ that is visible from a given sphere $V$ of radius $1/2$. An important observation here, which has also appeared implicitly in past work \cite{antoniadis2022online}, is that if $V$ has center $\vec{p}$, then the portion of $U$ that is visible from $V$ is also visible from the point $2 \vec{p}$. For an accompanying illustration, see Figure \ref{fig:cones}.

\begin{lemma}
Let $\vec{p}$ be a point with $|\vec{p}| \ge 1$. Let $V$ be a sphere of radius $1/2$ centered around point $p$. Let $Q_1$ denote the set of points $q \in U$ such that $q$ is visible from some point $r \in V$, and let $Q_2$ denote the set of points visible from point $2\vec{p}$. Then $Q_1 \subseteq Q_2$.
\label{lem:subpathvis}
\end{lemma}
\begin{proof}
Consider any point $q \in Q_1$.
As $q$ is visible from $r$, we have $\langle q,r\rangle \geq 1$. As $r \in V$, we can write $r=p + v$ for some $v$ with $|v|\leq 1/2$. Thus we have,
\[1 \leq \langle q,r\rangle = \langle q, p+v\rangle =  \langle q,p \rangle + \langle q,v\rangle \leq  \langle q,p \rangle  + |q||v| \leq \langle q,p \rangle + 1/2, \]
where the last step uses that $|q|=1$ and $|v|\leq 1/2$.
Thus $\langle q,p \rangle \geq 1/2$, or equivalently, $\langle 2p,q \rangle \geq 1$, which implies that $q$ is visible from the point $2p$. This completes the proof that $Q_1 \subseteq Q_2$.
\end{proof}



Combining the previous lemmas, we conclude that the only way for a path to see the entire sphere $U$ is if either (1) the path is very long, or (2) the path is, at some point, quite far away from the origin. 
\begin{lemma}
Consider a path $P$ that never surpasses distance $r$ from the origin, but that views all of $U$. We must have
\[|P| \ge e^{d / (8 r^2)} - 1.\]
\end{lemma}
\begin{proof}
Break $P$ into $\lceil |P| \rceil$ sub-paths $P_1, P_2, \ldots$, all but one of which have length $1$. Each $P_i$ is contained in a sphere $S_i$ of radius $1/2$ centered at the midpoint $p_i$ of $P_i$. By Lemma \ref{lem:subpathvis}, every point in $U$ visible from the subpath $P_i$ is also visible from the point $2p_i$. Since $2p_i$ is distance at most $2r$ from the origin, it follows by Lemma \ref{lem:pointvis} that the fraction of $U$ visible from $P_i$ is at most $e^{-d / (8 r^2)}.$
Since the $P_i$s collectively view all of $U$, we can conclude that $\lceil |P|  \rceil \ge e^{d / (8 r^2)}$, and thus that $|P| \ge e^{d / (8 r^2)} - 1$. 
\end{proof}

\begin{corollary}
Consider any path $P$ that views all of $U$. The path $P$ must either reach distance $\sqrt{d / (16 \log d)}$ from the origin or have length $|P| \ge e^{2 \log d} - 1 \ge \Omega(d^2)$. 
\label{cor:key}
\end{corollary}

We can now prove our main result.

\begin{theorem}
Any path $P$ that start at the origin and views all of $U$ must have length $|P| \ge \Omega(d^{3/2} / \sqrt{\log d})$. 
\end{theorem}
\begin{proof}
Let $\tau = \sqrt{\frac{d}{2} / (16 \log \frac{d}{2})}$. Suppose that $|P|$ has length $O(d^{3/2})$. By Corollary \ref{cor:key}, $P$ must reach distance $\tau$ from the origin -- let $P(t_1)$ be the first point at which this happens. Let $P_2$ be the path $P$ projected onto the $(d - 1)$-dimensional hyperplane orthogonal to $\vec{P(t_1)}$ (and going through the origin). Since $|P_2| \le |P|$, we know that $|P_2| \le O(d^{3/2})$, so by Corollary \ref{cor:key}, $P_2$ must reach distance $\tau$ from the origin -- let $P_2(t_2)$ be the first point at which this happens. Let $P_3$ be the path $P_2$ projected onto the $(d - 2)$-dimensional hyperplane orthogonal to $\vec{P(t_2)}$ (and going through the origin). Since $|P_3| \le |P_2|$, we know that $|P_3| \le O(d^{3/2})$, so by Corollary \ref{cor:key}, $P_3$ must reach distance $\tau$ from the origin -- let $P_3(t_3)$ be the first point at which this happens. Continuing like this, we can define $t_1, \ldots t_{\lfloor d / 2 \rfloor }$ and $P_1, \ldots, P_{\lfloor d / 2 \rfloor }$ so that each $P_i(t_i)$ is distance at least $\tau$ from the origin and so that each $P_{i + 1}$ is the path $P$ projected onto the $(d - i + 1)$-dimensional hyperplane (going through the origin) that is orthogonal to each of $P(t_1), P_2(t_2), \ldots, P_i(t_i)$. 

An important observation is that $t_1 \le t_2 \le t_3 \le \cdots \le t_{\lfloor d / 2 \rfloor }$. Indeed, if $t_{i + 1} < t_i$, then since $|P_i(t_{i + 1})| \ge |P_{i + 1}(t_{i + 1})| \ge \tau$, we would have that $t_i$ was not the first time at which $P_i$ reached distance $\tau$ from the origin, a contradiction. 

Also observe that the distance traveled by $P$ between time $t_i$ and time $t_{i + 1}$ is at least the distance traveled by $P_{i + 1}$ during that time interval. Since $P_{i + 1}(t_i) = \vec{0}$ and $|P_{i + 1}(t_{i + 1})| = \tau$, it follows that $P$ travels distance at least $\tau$ between $t_i$ and $t_{i + 1}$. The total length of $P$ is therefore at least
\[\sum_{i = 1}^{\lfloor d / 2 \rfloor} |P_i| \ge \lfloor d/2 \rfloor \cdot \tau \ge \Omega(d^{3/2} / \sqrt{\log d}). \qedhere\]
\end{proof}

\begin{corollary}
The optimal competitive ratio for the $d$-dimensional cow-path problem is $\Omega(d^{3/2} / \sqrt{\log d})$.
\end{corollary}

\bibliographystyle{plainurl} \bibliography{writeup}

\end{document}